\documentclass[10pt, twocolumn, journal]{IEEEtran}
\usepackage{nccmath} %Reduce formula size
\usepackage{graphicx}
\usepackage{amsmath}
\usepackage{amssymb}
\usepackage{setspace}
\usepackage[ruled,vlined,boxed]{algorithm2e}
\usepackage{color}
\usepackage{amsthm}
\usepackage{subcaption}

\newtheorem{theorem}{Theorem}[]
\newtheorem{lemma}[theorem]{Lemma}

\renewenvironment{proof}[1][Proof]{\begin{trivlist}
\item[\hskip \labelsep {\bfseries #1}]}{\end{trivlist}}

\begin{document}

%\begin{singlespace}
%\title{Construction and Optimality of Binary Constant Weight Codes with Combinatorial Code Parameters} 
\title{Combinatorial Entropy Encoding} 

%\author{\IEEEauthorblockN{Abu Bakar Siddique}\\}
%\IEEEauthorblockA{Department of Electrical Engineering\\ University of Engineering and Technology Lahore\\mabs239@gmail.com}
%Department of Electrical and Computer Engineering\\
%mabs239@gmail.com
%Telephone: (312) 545--7624
%}


%\author{\normalsize Abu Bakar Siddique \\
%Department of Electrical Engineering, RCET  \\  
%University of Engineering and Technology Lahore, Pakistan \\
%mabs239@gmail.com   % <-this % stops a space
%%\thanks{}% <-this % stops a space
%}

\author{\normalsize Abu~Bakar~Siddique* and Muhammad~Umar~Farooq\textdagger \\
\textdagger University of Engineering and Technology Lahore, Pakistan \\
*Department of Electrical Engineering, RCET  \\  
%Al-Khawarizmi Institute of Computer Science \\ 
University of Engineering and Technology Lahore, Pakistan \\
*abubakar@uet.edu.pk,~\textdagger mufarooq40@gmail.com   % <-this % stops a space
%\thanks{}% <-this % stops a space
}


\maketitle
\thispagestyle{empty} %To remove page number
%\IEEEpeerreviewmaketitle
%\end{singlespace}

%\doublespacing
%\onehalfspacing

\begin{abstract}
This paper proposes a novel entropy encoding technique for lossless data compression. Representing a message string by its lexicographic index in the permutations of its symbols results in a compressed version matching Shannon entropy of the message. Commercial data compression standards make use of Huffman or arithmetic coding at some stage of the compression process. In the proposed method, like arithmetic coding entire string is mapped to an integer but is not based on fractional numbers. Unlike both arithmetic and Huffman coding no prior entropy model of the source is required. Simple intuitive algorithm based on multinomial coefficients is developed for entropy encoding that adoptively uses low number of bits for more frequent symbols. Correctness of the algorithm is demonstrated by an example.
\end{abstract}

%==============================
%\section{Literature Survey}
%==============================
%\cite{burrows1994block}
%\cite{huffman1952method}
%\cite{witten1987arithmetic}
%\cite{cover1973enumerative} We provide an explicit scheme for calculating the index of any sequence in S according to its position in the lexicographic ordering of S
%\cite{ruskey2003combinatorial}
%\cite{takaoka19991} Constant time permutation generation with O(n) space complexity. n is the size of the emelemts in the message. It is not single index to permutation mapping. 
%\cite{sedgewick1977permutation} gives survey of permutation methods. Apart from recursive and interchange methods it lists leicographic methods such as 
%\cite{ruskey2003combinatorial}Loops and recursion to list the permutations in lexicographic order.
%\cite{williams2009loopless} Important. 

\section{Introduction}
Efficient storage and communication of ever growing digital content has kept researches busy designing new techniques for compressing this information efficiently. The ubiquitous presence of multimedia enabled communication devices produce enormous amount of data to be stored and communicated efficiently. Research in bio-informatics has introduced whole new dimension where large amount of DNA information has to be processed.  The handling, storage and transmission of these data is already beyond current available resources.

The science of data compression originated with Shannon's famous paper on information theory \cite{shannon2001mathematical} where he introduced entropy as measure of the information. 
Considering the source generating discrete set of symbols each of which is independent and identically distributed random variable, he showed that it is not possible to compress the source coding rate (average bits per symbol) below the Shannon entropy of the source and be able to reconstruct the original message without the loss of information. To define entropy suppose that a message is made up of $t$ alphabets in set $A={\alpha_1,\alpha_2,\cdots,\alpha_t}$. If the alphabet frequencies are ${f_1,f_2,\cdots,f_t}$ respectively with message length $n=f_1+f_2+\cdots+f_t$ then Shannon entropy of the message is

%\begin{equation}
%	\label{ShannonEntropy}
%	 H = \sum_{i=1}^t p_i \log_2 1/p_i \\
%	  =-\sum_{i=1}^t p_i \log_2 p_i
%\end{equation}

%%\begin{equation} 

  	\begin{eqnarray}
  		\label{eqn:ShannonEntropy}
  		 H&=&\sum_{i=1}^t p_i \log_2 1/p_i\\
  		   &=& -\sum_{i=1}^t p_i \log_2 p_i \nonumber
	\end{eqnarray}
%\end{equation} 

where $p_i = f_i /n$\\

Shannon, along with Fano, invented a variable length prefix coding scheme to compress data by assigning shorter codes to more frequent symbols and longer codes to rarely occurring symbols. A slightly improved version was presented by Huffman \cite{huffman1952method} that constructed the code assignment tree from leaves to the root, in contrast to the Shannon-Fano coding tree that built from root to the leaves. Huffman encoding is sometimes used directly to compress the data while more often it is used as final step to compress symbols after other compression algorithm has taken place. Tools like PKZIP, JPEG, MP3, BZIP2 make use of Huffman coding. For efficient compression, in the original Huffman encoding, it is required that the character probabilities be known at the start of the compression process. This can come either from an accurate estimate of the information source or from statistical pre-processing of the data before putting it to compression algorithm. This encoding approaches the Shannon limit when the data size, $n$, is large and the character probabilities are powers of $1/2$. It is least efficient when the data consists of binary symbols since each Huffman symbol is integer number of bits in size. Such is the problem with other prefix codes such as Golomb, Elias and unary coding as only rarely symbol probabilities obey this condition in practice. 

To solve the problem of inefficient Huffman encoding due to integer bits another statistical based data compression scheme known as arithmetic encoding \cite{witten1987arithmetic}was introduced. Though theoretically promising, its use remained initially limited because of the patent issues. It encodes all the data into one fractional number, represented by bits no more than the Shannon limit. It's beauty is that it works even for binary alphabet, assigning fractional bits to symbol. One drawback is that the complete message is processed with no inter-alphabet boundaries. Since digital computers are more comfortable working with integers, the fractional ranges have to be frequently converted to integer ranges in practice. Some of these limitations were solved by range coding \cite{martin1979range}. However these are still costly because of multiplication operations.

 The author was unable to find a statistical encoding scheme that mentions or uses combinatorics to compress the data. Previously the source statistics are defined in terms of the symbol probabilities, the fractional numbers. Our point of view is to consider a message as one permutation of its constituent symbols. Given the symbols and their frequencies in the message, we can represent the input data with a small integer denoting the lexicographic position of the message in list of permutations of its symbols. This index takes far fewer bits than the original message and achieves perfect compression. Contribution of this paper is also to develop an logarithm that calculates the dictionary index of a message as permutation of its constituent symbols with given frequencies for message encoding and vice versa for decoding. The index calculated this way remains below the Shannon entropy of the message.  

\begin{figure*}[!h]
\centering
\includegraphics[width = 4in]{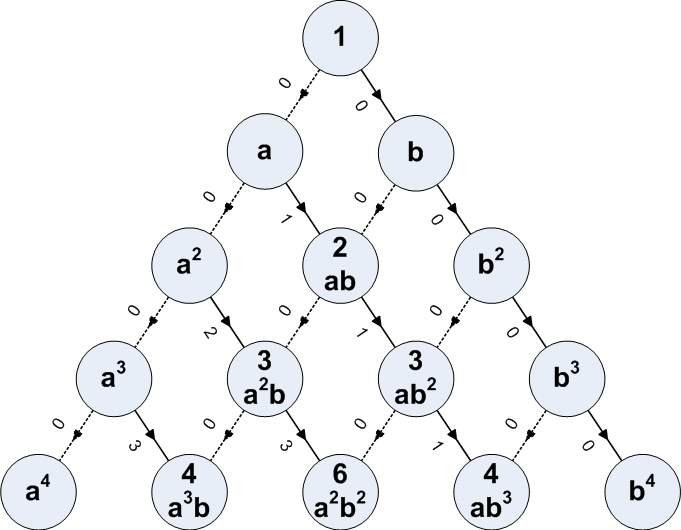}
\caption{Binary source data compression using the Pascal's triangle. In this directed acyclic graph, node pascal coefficient gives the number of paths from the root node.  Powers of \textbf{a} and \textbf{b} tell the frequencies of these binary symbols in the message and correspond to number of left and right turns, respectively, taken to reach from root node. Starting at root node \textbf{1} and traversing the graph by taking a left turn on each \textbf{a} and right turn for each \textbf{b}, the sum of edge weights gives lexicographic order of the message. For example the message \textbf{``aabb''} would lead to node $\mathbf{6a^2b^2}$ and maps to lexicographic index $0+0+2+3=\mathbf{5}$. The message \textbf{``bbaa''} would map to $0+0+0+0=\mathbf{0}$.}
\label{fig:iterative}
\end{figure*}

The paper presents a data compression technique that to the best knowledge of the author has not been presented before. The message compression is achieved by combinadics, a mixed radix numbering system. The algorithm treats the input message as one permutation of its constituent symbols. Like both Huffman encoding and Arithmetic encoding an alphabet frequency list is needed for decompression of the message. The compression however does not need pre-processing or statistical modeling of the information source. Rather the frequency of each letter encountered at the time of compression is incremented and the compression process continues at the mean time. Compression has linear time complexity  $\mathbf{O}(n)$ while decompression is $\mathbf{O}(t\cdot n)$. 

 In this paper  $n$ represents the message length, $t$ represents the alphabet size. The set of alphabets $A=(\alpha_1,\alpha_2,\cdots,\alpha_t)$ is replaced with set $(0,1,2,\cdots, t-1)$ to use with message encoding examples. Each occurrence of an alphabet is termed as a message symbol. There are $n$ symbols in $n$ length message.

%\begin{equation}
%\label{EdgeWeight}
%W_{ \binom{n-1}{f_0,\cdots,f_t-1,\cdots,f_m} \rightarrow \binom{n}{f_0,\cdots,f_t,\cdots,f_m}} = \sum_{i=0}^{t-1} \binom{n-1}{f_0,\cdots,f_i-1,\cdots,f_m}
%\end{equation}

\section{Compression}

\begin{figure*}
\centering
  \begin{subfigure}[b]{.3\textwidth}
   \includegraphics[width=\textwidth]{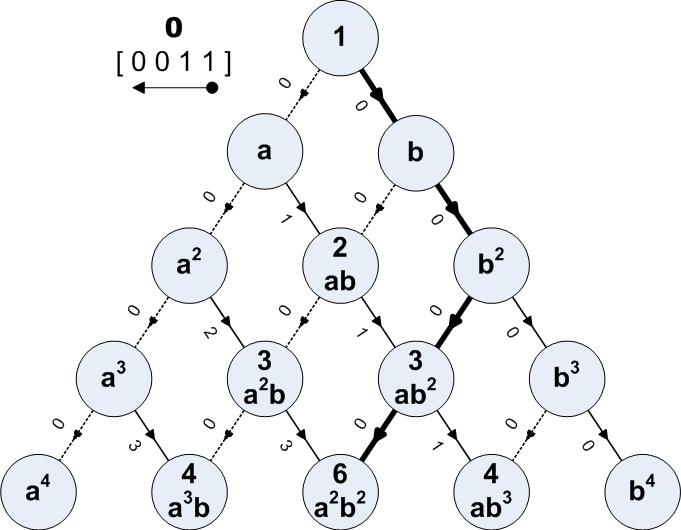}
   \label{subfig:b} 
   \caption{Input  $[0011]$ mapped to [0], \\ Path traversal from root: RRLL,\\Sum of edge weights: 0+0+0+0=\textbf{0}.}
  \end{subfigure}
 ~
  \begin{subfigure}[b]{.3\textwidth}
   \includegraphics[width=\textwidth]{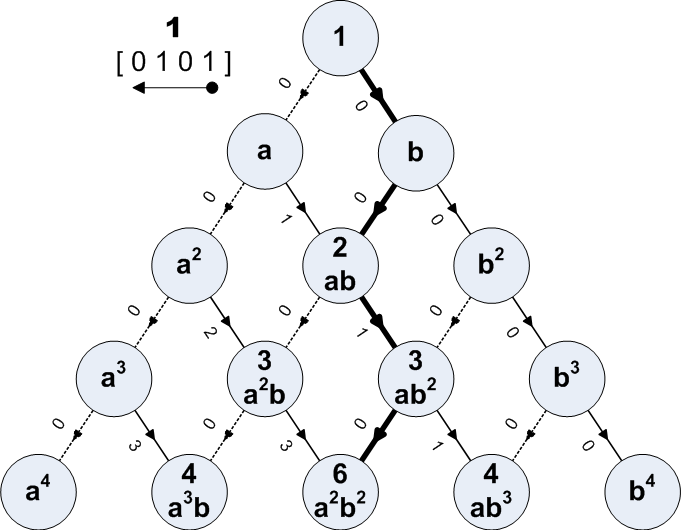}
   \label{subfig:b} 
   \caption{Input  $[0101]$ mapped to [1],
   		 \\ Path traversal from root: RLRL,
   		  \\Sum of edge weights: 0+0+1+0=\textbf{1}}  
\end{subfigure}
 ~
  \begin{subfigure}[b]{.3\textwidth}
   \includegraphics[width=\textwidth]{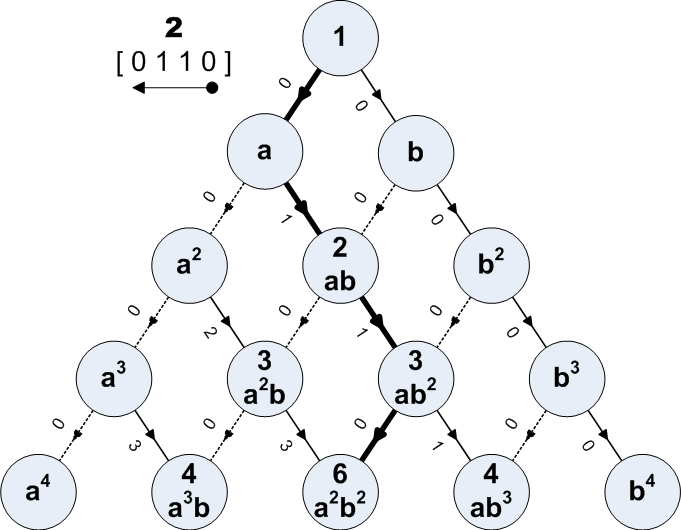}
   \label{subfig:b} 
   \caption{Input  $[0110]$ mapped to [2],
   		 \\ Path traversal from root: LRRL,
   		  \\Sum of edge weights: 0+1+1+0=\textbf{2}}  
  \end{subfigure}
	
 \vspace{2em}

 \centering
  \begin{subfigure}[b]{.3\textwidth}
   \includegraphics[width=\textwidth]{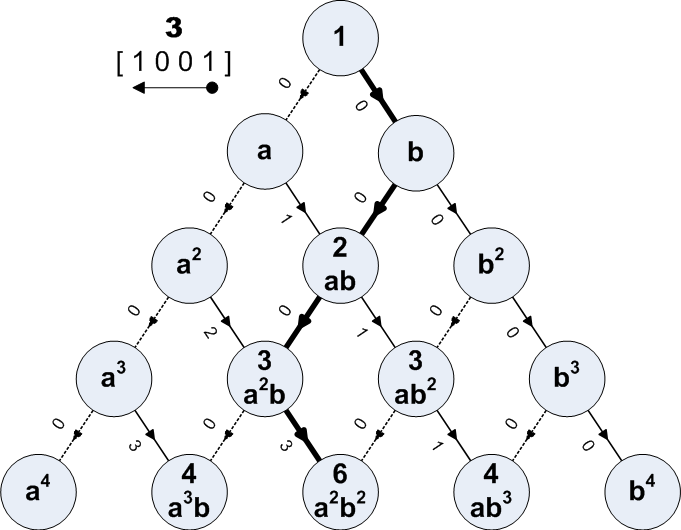}
   \label{subfig:b} 
   \caption{Input  $[1001]$ mapped to [3],
   		 \\ Path traversal from root: RLLR,
   		  \\Sum of edge weights: 0+0+0+3=\textbf{3}}  
  \end{subfigure}
 ~
  \begin{subfigure}[b]{.3\textwidth}
   \includegraphics[width=\textwidth]{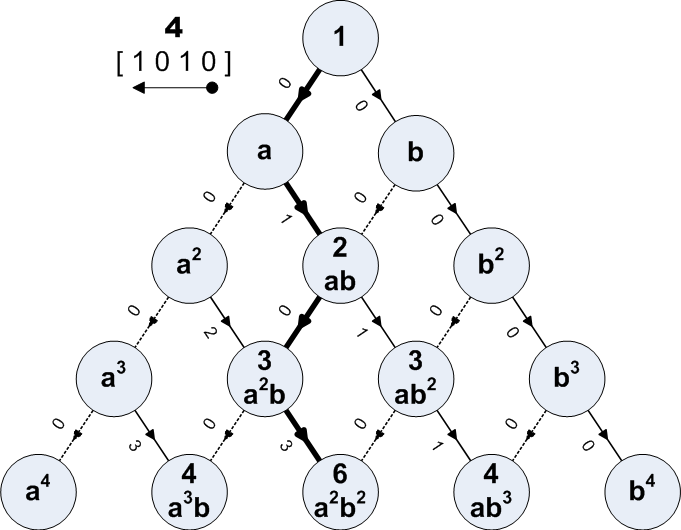}
   \label{subfig:b} 
   \caption{Input  $[1010]$ mapped to [4],
   		 \\ Path traversal from root: LRLR,
   		  \\Sum of edge weights: 0+1+0+3=\textbf{4}}  
  \end{subfigure}
 ~
  \begin{subfigure}[b]{.3\textwidth}
   \includegraphics[width=\textwidth]{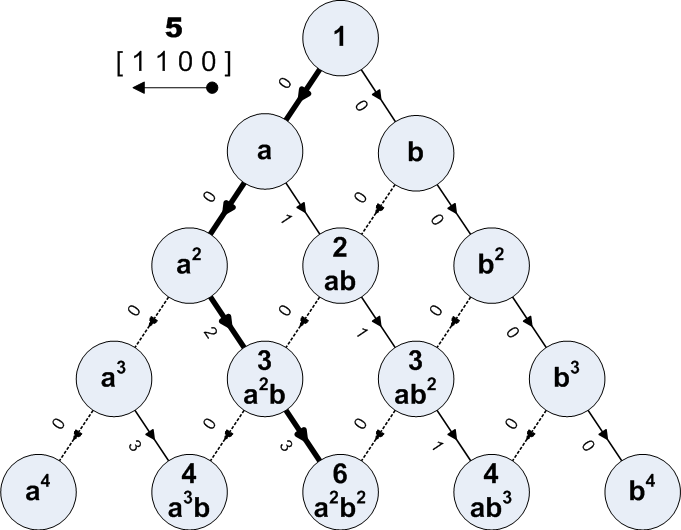}
   \label{subfig:b} 
   \caption{Input  $[1100]$ mapped to [5],
   		 \\ Path traversal from root: LLRR,
   		  \\Sum of edge weights: 0+0+2+3=\textbf{5}}  
  \end{subfigure}

 \caption{Encoding for binary messages of length $n=4$. Traveresal starts at root node. For the input stream we starts from Least Significant Bit (LSB). If a one is encountered in the input stream Right (R) turn is taken in graph traversal and Left (L) turn for a zero. Because all six streams have 2 zeroes (a's) and 2 ones (b's), we end up at node $\mathbf{6a^2b^2}$ in each case. Unlike Huffman/Arithmatic/Range encoding, no prior source model or alphabet probablity knowledge is requiered. Rather respective alphabet frequency counters are incremented for each new symbol arrival.}
\label{fig:binaryExamples}
\end{figure*}

The data compression is treated as combinatorics problem in this paper. When alphabet frequences in the message are not known then for a message length $n$ and alphabet size $t$ the number of all possible messages is $n^t$. For example if message length is six $(n=6)$ and three alphabets $(t=3)$ are used from set $A = (0,1,2)$, the number of possible messages would be $3^5 = 243$. The binary bits required to represent the message can be calulated as $\log_2 (3^5)= 9.5098$. With 6 symbols long message, the average information per symbol is $9.5098/6=1.5850$ bits/symbol. The message can not be decodeably represented with fewer than 10 bits if no other statistics information is available. 

If the probabilites of the alphabets are known in advance then the number of bits requiered to span all possible messages would be much lower as given by Shannon source coding limit. If the alphabets $(0,1,2)$ appear in the message with probabilities ${3/6,1/6,2/6}$ respectively, the optimum symbol size by Shannon source coding theorem in Equation-\ref{eqn:ShannonEntropy}

\begin{equation} 
	\label{bananaEntropy}
  	\begin{array}{lcl}
  		 H&=&-\left( 3/6 \log_2 (3/6)+1/6 \log_2 (1/6)+2/6 \log_2 (2/6) \right) \\
  		   &=& 1.4591 \text{Bits/Symbol}\\
	\end{array}
\end{equation} 

Reducing the average information per symbol to the Shannon limit is the goal of every compression method. 

Compression ratio is a measure of efficiency for a compression method. It is the ratio

\begin{equation}
	\text{Compression Ratio} = \frac{\text{Uncompressed Size}}{\text{Compressed Size}}
\end{equation}
A closely related measure quanifying the percentage amount of space saved by the compression method is
\begin{equation}	
	\text{Space Saving (\%age)} = (1-\frac{\text{Compressed Size}}{\text{Uncompressed Size}})\times 100 
\end{equation}
The compression ratio in our case is $1.5850/1.4591 = 1.0863$ percent or 7.94\% reduction in the amount of space needed to store the mesage. 

This paper treats the compression as a combinatorial problem. We evaluate the total permutations of the message given the alphabets used and their respective occurance frequencies in the messsage. The same statistical information is used by Shannon, Huffman and arithmetic encoding as alphabet probability of occurance in the message. Instead of working with alphabet probabilities we work with the alphabet frequencies. Taking the previous example further,  calculate the number of permutations of 6-symbol message formed by three alphabets, appearing in the message with frequencies three, one and two say. The permutations calculation is the same as combinatorial problem of placing 6 balls: say three red, two green and one blue, in 6 urns. The enumerations are calculated using the multinomial formula

%\begin{equation}
\begin{eqnarray}
	\label{eqn:multinomial}
P(n,f_1,f_2,\cdots,f_t)&=& \binom{n}{f_1,f_2,\cdots,f_t} \\
&=& \frac{n!}{f_1! f_2!\cdots f_t!} \nonumber
\end{eqnarray}
%\end{equation}
where 
$$ n = f_1+f_2+\cdots+f_t$$

The number of permutations for the example using Equation-\ref{eqn:multinomial} would be

\begin{eqnarray}
\binom{6}{3,1,2} &=& \frac{6!}{3!\cdot 1!\cdot 2!} \\
&=& \frac{6\cdot 5\cdot 4\cdot 3\cdot 2\cdot 1}{(3\cdot 2\cdot 1)\cdot (1)\cdot (2\cdot 1)}\nonumber \\
&=&60 \nonumber
\end{eqnarray}

A unique ID can be attributed to each of these 60 permutations using  $(  \log_2(60) ) = 5.9069$. Because a binary number can not have fractional number of bits therefore 6-bits would suffice to uniquely identify each permutation.  As a comparasion, this is much smaller than the Shannon length  $6\times 1.4591=8.7546=$ 9-bits. It achieves the space saving of  ratio $1/1.5850=63.09 \%$. 

It essentially means that we need a mechanism to establish bijection between the message permutation and its order in the list of all the permutations of given alphabet with given frequency counts. That number represented in biary format  would correspond to the origional message. The best selection would be the entry number in the dictionary of all permutations of the origional message. Algorithms exist that can find the next permutation in lexicographic order given the present permutation. However, to the best of the author's knowledge, no algorithm has been presented in literature that can find the permutation lexicographic position given the message or vice versa. This paper presents origional algorithm to map a message to and from its lexicographic order with simple method. A mixed radix number system consisting of multinomial coefficients is developed to perform this mapping.

The reader might get skeptic of getting a bit representation smaller than the Shannon limit in our toy example. The combinatorial compression would approaches the Shannon limit as $n$ approaches $\infty$ for uniform alphabet frequencies but remaining always below the Shannon limit.

For a comparasion the Shannon formula can be re-arranged to enumerate the corresponding patterns for given alphabet probabilities $p_1,p_2,\cdots,p_t$.
\begin{align}
H &= p_1\log_2 (1/p_1)+p_2\log_2 (1/p_2)+\cdots+p_t\log_2 (1/p_t) \\
    & = \log_2(\frac{1}{p_1^{p_1}\cdot p_2^{p_2}\cdots p_t^{p_t} }) \nonumber 
\end{align}
As $H$ is the average length per symbol, $n\cdot H$ would give the total message length and $2^{nH}$ will give the total number of possible messages with $nH$ bits.

\begin{align}
  n\cdot H &=n\cdot  \log_2(\frac{1}{p_1^{p_1}\cdot p_2^{p_2}\cdots p_t^{p_t} } )\\
   2^{n\cdot H} &=  (\frac{1}{p_1^{p_1}\cdot p_2^{p_2}\cdots p_t^{p_t} } )^n  \\
      2^{n\cdot H} &=  (\frac{1}{p_1^{n\cdot p_1}\cdot p_2^{n\cdot p_2}\cdots p_t^{n\cdot p_t} } ) 
\end{align}

Using the relation between alphabet frequency and its probability $p_i = f_i/n$

\begin{align}
      2^{n\cdot H} &=  \frac{1}{(f_1/n)^{f_1}\cdot (f_2/n)^{f_2}\cdots (f_t/n)^{f_t} } \\
      2^{n\cdot H} &=  \frac{n^{f_1}\cdot n^{f_2}\cdots n^{f_t}}{(f_1)^{f_1}\cdot (f_2)^{f_2}\cdots (f_t)^{f_t} } \\
      2^{n\cdot H} &=  \frac{n^{f_1+f_2+\cdots+f_t}}{(f_1)^{f_1}\cdot (f_2)^{f_2}\cdots (f_t)^{f_t} } \\      
      2^{n\cdot H} &=  \frac{n^n}{f_1^{f_1}\cdot f_2^{f_2}\cdots f_t^{f_t} }  \label{eqn:3} 
\end{align}

It is trivial to prove that Equation-\ref{eqn:3} is always greater than Equation-\ref{eqn:multinomial} for all $n$ and $f_i$.

\begin{lemma}
\label{lemma:lemma1}
\begin{eqnarray}
 \frac{n^n}{f_1^{f_1}\cdot f_2^{f_2}\cdots f_t^{f_t} }  > \frac{n!}{f_1!\cdot f_2! \cdots f_t!}
\end{eqnarray}
 for all $n$ and $f_i$. Where $\sum_{i=1}^t f_i = n$
\end{lemma}
 
\begin{proof}
The inequality can be rearranged to the form
\begin{eqnarray}
 n^n  > \frac{n!}{f_1!\cdot f_2! \cdots f_t!} {f_1^{f_1}\cdot f_2^{f_2}\cdots f_t^{f_t} }\\
 (f_1+\cdots+f_t)^n  > \binom{n}{f_1,\cdots,f_t!} {f_1^{f_1}\cdots f_t^{f_t} } \label{eqn:HProof}
\end{eqnarray}

Clearly (\ref{eqn:HProof}) is true because right hand side is a term in the expansion of left hand side. Q.E.D.

\end{proof}

\begin{theorem}

Let a source stream of independent identically distributed random variables  of length $n$, of which each symbol comes from the set $A=\left\{ \alpha_1,\alpha_2,\cdots,\alpha_t \right\}$ with frequencies $(f_1,f_2,\cdots,f_t)$ respectivly. 
It is possible to compress the origional message to $\log_2 \binom{n}{f_1,f_2,\cdots,f_t}<nH$ bits with lossless reconstruction for all $n$ and $f_i$.
\end{theorem}

\begin{proof}
    Given the alphabet frequencies $(f_1,f_2,\cdots,f_t)$, due to basic combinatorics there can not be more that $\binom{n}{f_1,f_2,\cdots,f_t}$ permutations of the to-be-compressed-message. If a bijection is established between each permutation of the message and a number representing the lexicographic position of the permutation, no more than $\log_2 \binom{n}{f_1,f_2,\cdots,f_t}$ would be needed to represent any of the permutations. This lexicographic position along with alphabet frequencies would accurately regenerate the origional message. Due to Lemma-\ref{lemma:lemma1}   $1/n \log_2 \binom{n}{f_1,f_2,\cdots,f_t}$ is always smaller than $H$. The difference is more prominent for small $n$. It is also substantial when one or two alphabets occur more frequently than others. For large $n$ and uniform $f_i$, the lemma approximates equality.
\end{proof}

\section{Data Compression for Binary Source}
The data compression for combianatoric encoding is first explained by a binary source with alphabet $\mathbf{A}=\left\{a=0,b=1\right\}$. Figure-\ref{fig:binaryExamples} shows a few examples of combinatorial encoding. The figure proves the bijection of message permutations and its lexicographic order using the Pascal's triangle. Each node in the directed acyclic graph can be reached from root in $\binom{n}{f_0,f_1}a^{f_0}b^{f_1}$ ways, where $n$ is the level of the node in the graph and $f_0$ and $f_1$ are the number of left and right turn needed to reach the node, respectively. Here the binomial coefficient is written in more generic form using $\binom{n}{k}=\binom{n}{k,n-k}$. It will help extend the discussion to multinomial coefficients. The edge wigths between two nodes are given by the eqation

\begin{algorithm}[] 
%\small
\caption{Combinatorics Encoder  for Binary Alphabet} \label{algo:encoder}
\KwData{Source message $S$ of length $n$}
\KwOut{Lexicographic index: $L_i$,  Frequencies: $f_0,f_1$}
$L_i\longleftarrow 0$ \;
$f_0\longleftarrow 0$ \;
$f_1\longleftarrow 0$ \;
%\For{$i = 0:n-1$}
\For{$i\longleftarrow 0$ \KwTo $n-1$}
{
\eIf{$S[i]==1$ }
{
	\If{$i>f_1$}{$L_i \longleftarrow L_i + \binom{i}{f_1}$}
	$f_1\longleftarrow f_1+1$\;
}{$f_0\longleftarrow f_0+1$ \; }
}
\end{algorithm}

\begin{eqnarray}
W_{\text{Left Turn}} \left\{ \binom{n-1}{f_0-1,f_1}\rightarrow\binom{n}{f_0,f_1} \right\}&=& \text{Zero}\\
W_{\text{Right Turn}}\left\{ \binom{n-1}{f_0,f_1-1}\rightarrow\binom{n}{f_0,f_1} \right\}&=& \binom{n-1}{f_0-1,f_1}   \nonumber
 \end{eqnarray} 

Each path traversed with above edge weights will give a unique path cost, refered to here as lexicographic number of the path.

 Compression is achieved by calculating the lexicographic order of the input message. For example if the input message is [11011100101], its lexicographic order or the compressed representation can be computed using Algorithm-\ref{algo:encoder}:
%1,5,7,10 [10010100010]
%$L_i =\binom{10}{4}+\binom{7}{3}+\binom{5}{2}+\binom{1}{1}=210+35+10+1=256 $
\begin{eqnarray}\nonumber
L_i &=&\binom{10}{7} +\binom{9}{6} +\binom{7}{5} +\binom{6}{4} +\binom{5}{3} +\binom{2}{2} +\binom{0}{1}\\ \nonumber
&=&120+84+21+15+10+1+0 \\ \nonumber
&=&251 \nonumber
\end{eqnarray}
The number 251 can be represented with 8-bits, a saving of three bits. For decompression we will need the information $(L_i=251,~f_1=7,~f_0=4)$.

For easy hand calculations the above algorithm can be summarized as follows:

\begin{itemize}
\item Index message bits from zero starting at Least Significant Bit (LSB) of the message. Let $i$ be the bit index.
\item Start a counter $j$ from LSB side and increment wheneveve a ONE is encountered. If LSB start with ONE then first value of $j$ would be one.
\item Compute lexicographic order of the message as $\sum_{i=0}^{n-1}b_i\binom{i}{j}$, where $b_i$ is the $i^{th}$ bit and $\binom{i}{j}=0$ for $i<j$
\end{itemize}

It should be noted lexicographic order remains unchanged for any number of zeros on left of message MSB. It helps in representing the bits in more compact form. However zeros to the right of LSB are significant and they effect the lexicographic order.

For lossless decompression of the message we need the lexicographic order of the message with given frequencies of the alphabets.

\begin{figure*}
\centering
%  \begin{subfigure}[b]{.3\textwidth}
   \includegraphics[width=.8\textwidth]{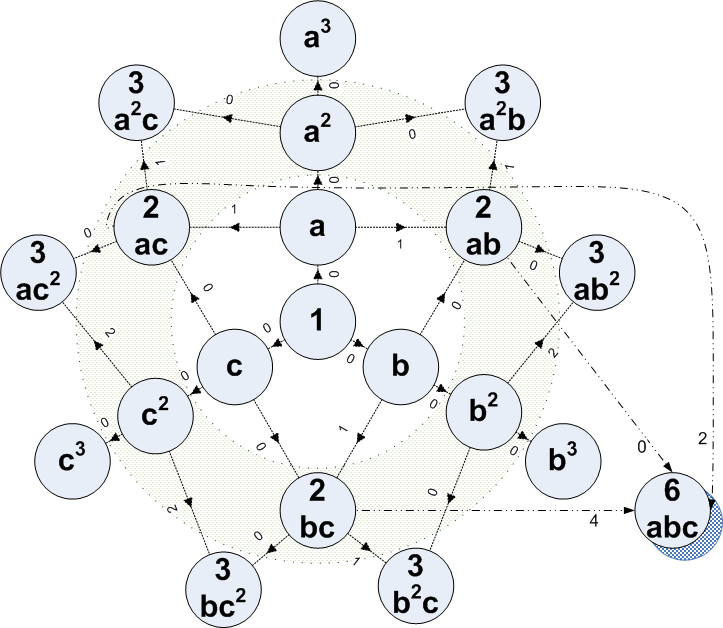} 
   \caption{Messages constructed by t-alphabet are encoded using acyclic directed graph of Pascal simplex in t-dimensions. For three alphabet we get a tetrahedron. }
     \label{fig:trinomial}
  \end{figure*}

%\begin{equation} \label{eq:ubound1}
%%\binom{n-1}{f_0,\cdots,f_t-1,\cdots,f_m} \rightarrow \binom{n}{f_0,f_1,\cdots,f_m} = 
%%	\left\{
%	 \begin{array}{llll}
%131 &=& \mathbf{1} \cdot ^9C_6 + 47&\longrightarrow 131 = 1 x 84 + 47 \vspace{1em}\\
%47 &=& \mathbf{0} \cdot \binom{8}{5} + 47&\longrightarrow 47 = 0 x 56 + 47\\
%47 &=& \mathbf{1} \cdot \binom{7}{5} + 26&\longrightarrow 47 = 1 x 21 + 26\\
%26 &=& \mathbf{1} \cdot \binom{6}{4} + 11&\longrightarrow 26 = 1 x 15 + 11\\
%	 \end{array} %\right..
%\end{equation}

\begin{algorithm}[] 
%\small
\caption{Data Reconstruction Algorithm} \label{algo:decoder}
\KwIn{Lexicographic index: $L_i$,  Frequencies: $f_0,f_1$}
\KwOut{Data array $S[~]$}
\While{$(f_0+f_1)\geq 0$}
%\For{$n\longleftarrow (f_0+f_1-1)$ \KwTo 0}
{
$n \longleftarrow (f_0+f_1-1)$ \;
%\eIf{$0 < f_1 < m_i$ }
%{
%$y = \binom{m_i-1}{k_i}$
%}{
%$y=0$
%}
\eIf{$L_i\geq \binom{n}{f_1}$}
{
$L_i \longleftarrow L_i - \binom{n}{f_1}$ \;
$S[n] \longleftarrow 1$ \;
$f_1 \longleftarrow f_1 -1$\;
}
{
$S[n] \longleftarrow 0$\;
$f_0 \longleftarrow f_0 -1$\;
}
}
\end{algorithm}

Using Algorithm-\ref{algo:decoder} the compressed data $(L_i=251,~f_1=7,~f_0=4)$ can be decompressed as follows.

\begin{eqnarray*} \nonumber
251 &= & \mathbf{1}  \cdot \binom{10}{7} + 131  \longrightarrow  \mathbf{[1\cdots]} \\ \nonumber
131 &= & \mathbf{1}  \cdot \binom{9}{6} + 47  \longrightarrow  \mathbf{[11\cdots]} \\ \nonumber
47 &= & \mathbf{0}  \cdot \binom{8}{5} + 47  \longrightarrow  \mathbf{[110\cdots]} \\ \nonumber
47 &= & \mathbf{1}  \cdot \binom{7}{5} + 26  \longrightarrow  \mathbf{[1101\cdots]} \\ \nonumber
26 &= & \mathbf{1}  \cdot \binom{6}{4} + 11  \longrightarrow  \mathbf{[11011\cdots]} \\ \nonumber
11 &= & \mathbf{1}  \cdot \binom{5}{3} + 1  \longrightarrow  \mathbf{[10111\cdots]} \\ \nonumber
1 &= & \mathbf{0}  \cdot \binom{4}{2} + 1  \longrightarrow  \mathbf{[101110\cdots]} \\ \nonumber
1 &= & \mathbf{0}  \cdot \binom{3}{2} + 1  \longrightarrow  \mathbf{[1011100\cdots]} \\ \nonumber
1 &= & \mathbf{1}  \cdot \binom{2}{2} + 0  \longrightarrow  \mathbf{[10111001\cdot\cdot]} \\ \nonumber
0 &= & \mathbf{0}  \cdot \binom{1}{1} + 0  \longrightarrow  \mathbf{[101110010 \cdot]} \\ \nonumber
0 &= & \mathbf{1}  \cdot \binom{0}{1} + 0  \longrightarrow  \mathbf{[1011100101]}  \nonumber
\end{eqnarray*}

\section{Data Compression For t-Alphabet Source}
The data compression for multi alphabet source is achieved on the same principes as used by the binary source. An acyclic directed graph can be constructed in this case also. The major difference is that there are now $t$ outwards paths from each node, connecting to next $t$ nodes. Edge weights are weighted multinomial coeffiecients instead of binomial coefficients. If we have the alphabet as $(0,1,2,\cdots,t-1)$, then each node in the directed graph has the data string $\binom{n}{f_1,\cdots,f_t}\alpha_1^{f_1}\cdots\alpha_t^{f_t}$. Figure-\ref{fig:trinomial} shows the graph for three alphabet case.

The nodes on each level have labels
\begin{eqnarray}
\sum_{f_1+f_2+\cdots+f_t=n}\binom{n}{f_1,f_2,\cdots,f_t} \alpha_1^{f_1}\alpha_2^{f_2}\cdots\alpha_t^{f_t}\
\end{eqnarray}

The weight of the directed edge connecting two nodes $_{\binom{n-1}{f_1,\cdots,f_k-1\cdots,f_t} \alpha_1^{f_1}\cdots \alpha_k^{f_k-1}\cdots\alpha_t^{f_t} \rightarrow \binom{n}{f_1,\cdots,f_k,\cdots,f_t} \alpha_1^{f_1}\cdots \alpha_k^{f_k}\cdots\alpha_t^{f_t} }$ is given by

\begin{eqnarray}\medmath
W_{  \binom{n-1}{f_1,\cdots,f_k-1\cdots,f_t} \rightarrow \binom{n}{f_1,\cdots,f_k,\cdots,f_t} }  = \sum_{j=1}^{k-1}\binom{n-1}{f_1,\cdots,f_j-1,\cdots,f_t}
%W_{ \left\{  \binom{n-1}{f_1,\cdots,f_k-1\cdots,f_t} \rightarrow \binom{n}{f_1,\cdots,f_k,\cdots,f_t} \right\} }  = \sum_{j=1}^{k-1}\binom{n-1}{f_1,\cdots,f_j-1,\cdots,f_t}
\end{eqnarray}

\begin{figure*}
\centering
  \begin{subfigure}[b]{.4\textwidth}
   \includegraphics[width=\textwidth]{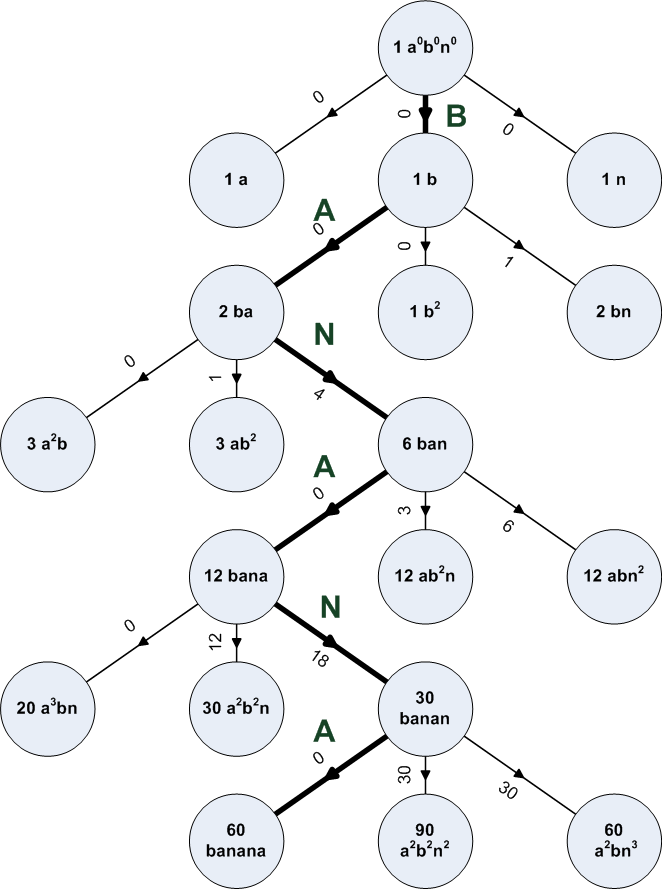}
   \label{subfig:b} 
   \caption{$[B A N A N A]$ mapped to lexicographic index [22], \\ Pascal tetrahedron traversal looking at next nodes.}
  \end{subfigure}
 ~
  \begin{subfigure}[b]{.4\textwidth}
   \includegraphics[width=\textwidth]{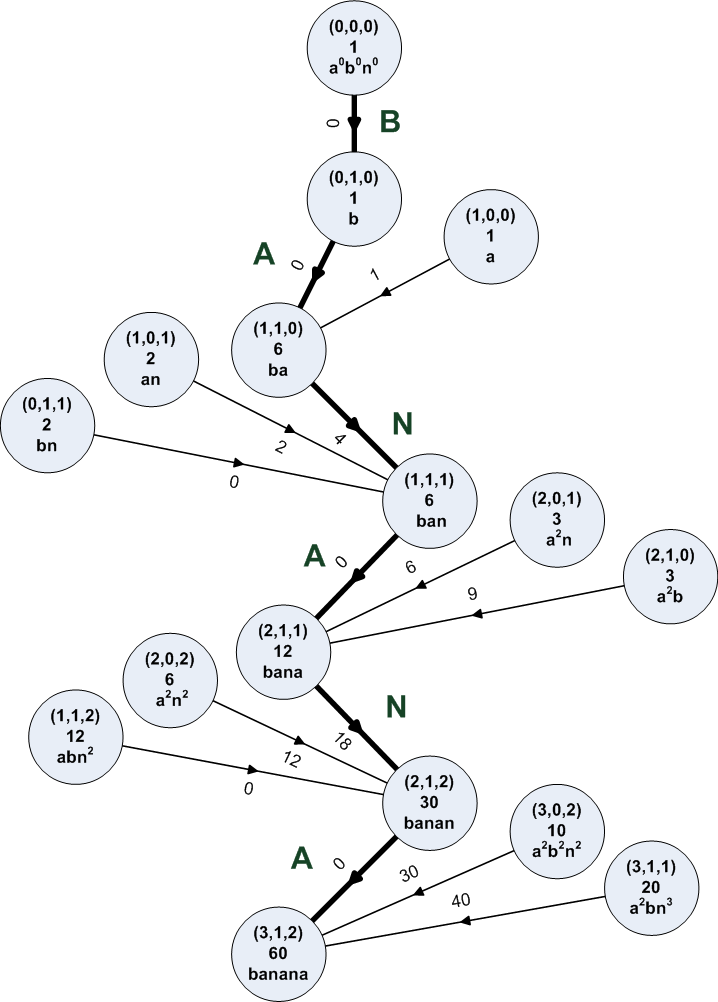}
   \label{subfig:b} 
   \caption{$[B A N A N A]$  mapped to lexicographic index [22], \\ Pascal tetrahedron traversal looking at previous  nodes.}  
\end{subfigure}

 \caption{Input stream $[B A N A N A]$ is encoded by traversal in directed acyclic graph of Pascal's tetrahedron as there are three alphabets. compressed to [10110]. $6$ edges are traversed, taking directions according to input alphabet.}
\label{fig:binaryExamples}
\end{figure*}

\subsection{Encoding Example for Tertiary Alphabet:  \textbf{B A N A N A}}

All possible permutations of the letters in BANANA are
$$P(n,f_1,f_2,f_3)= \binom{6}{3,1,2} = \frac{6!}{3! 2! 1!}= \frac{6\cdot 5\cdot 4\cdot 3\cdot 2\cdot 1}{3\cdot 2\cdot 1\cdot 1\cdot 2\cdot 1}=60$$

So there are merely 60 different permutations. If we consider dictionary entry number of each of these permutations, any permutation can be represented by a binary number of no more than 6-bits length. That corresponds to one bit per symbol which is considerable compression as compared to 2-bits/symbol required for ordinary binary encoding.

% If we check the message entropy using Shannon's formula:

\begin{table*}[t]
%\small
\centering 
\caption{All the permutations of word \textbf{B A N A N A}: $\lceil \log_2\binom{6}{3,1,2}\rceil=6$ bit long binary code along with knowledge of character frequencies $(f_A=3,f_B=1,f_N=2)$  is sufficient to uniquely represent any permutation.}
\label{tab:ExampleCode} \small
\begin{tabular}{l|c|r||l|c|r||l|c|r}
 \hline \hline %\vspace{1pt}
Lexicogra- & Permutation & Code & Lexicogra- & Permutation & Code & Lexicogra- & Permutation & Code\\
phic Index &  &  & phic Index &  &  & phic Index &  & \\
\hline\hline
 &  &  &  &  &  &  &  & \\
0& nnbaaa &1&20& nabana &10101&40& nbaaan & 101001\\
1& nbnaaa &10&21& anbana &10110&41& bnaaan & 101010\\
2& bnnaaa &11&22& \textbf{banana} &10111&42& nabaan & 101011\\
3& nnabaa &100&23& abnana &11000&43& anbaan & 101100\\
4& nanbaa &101&24& naabna &11001&44& banaan & 101101\\
& & && & && & \\
5& annbaa &110&25& anabna &11010&45& abnaan & 101110\\
6& nbanaa &111&26& aanbna &11011&46& naaban & 101111\\
7& bnanaa &1000&27& baanna &11100&47& anaban & 110000\\
8& nabnaa &1001&28& abanna &11101&48& aanban & 110001\\
9& anbnaa &1010&29& aabnna &11110&49& baanan & 110010\\
& & && & && & \\
10& bannaa &1011&30& nnaaab &11111&50& abanan & 110011\\
11& abnnaa &1100&31& nanaab &100000&51& aabnan & 110100\\
12& nnaaba &1101&32& annaab &100001&52& naaabn & 110101\\
13& nanaba &1110&33& naanab &100010&53& anaabn & 110110\\
14& annaba &1111&34& ananab &100011&54& aanabn & 110111\\
& & && & && & \\
15& naanba &10000&35& aannab &100100&55& aaanbn & 111000\\
16& ananba &10001&36& naaanb &100101&56& baaann & 111001\\
17& aannba &10010&37& anaanb &100110&57& abaann & 111010\\
18& nbaana &10011&38& aananb &100111&58& aabann & 111011\\
19& bnaana &10100&39& aaannb &101000&59& aaabnn & 111100\\
 &  &  &  &  &  &  &  & \\
\hline
\end{tabular}
\end{table*}

%\begin{eqnarray}
%L_{max} =\log_2 \binom{n}{f_A,f_B,f_N}\\
%E[L] = \sum_{i=0}^{L_{max}-1} 
%E[L] = 1*1 + 
%\end{eqnarray}

\section{Conclusion}

An entropy encoding technique based on permutation index representation of the message string is presented. The encoding and decoding algorithms work by assigning integer weights to symbol place value based on multinational coefficients. Encoding process use only addition and subtraction and unlike arithmetic or range encoding require no multiplication or division operations. With pre-calculated binomial or multinomial coefficients the method achieves the compression accuracy of arithmetic coding at the operation speed of Huffman coding. The binary messages need binomial coefficients as digit place values while multinomial coefficients are used for higher cardinality alphabet sets.

\bibliographystyle{IEEEtran}
%\bibliography{references_CWC}
\bibliography{reference}
\end{document}